\documentclass{amsart}
\usepackage{amsmath}
\usepackage{amssymb}
\usepackage{amsbsy}
\usepackage{amsfonts}
\usepackage{epsfig}
\usepackage{mathrsfs}
\usepackage{pict2e}
\usepackage{graphicx}
\usepackage{multirow}
\usepackage{epstopdf}
\usepackage{verbatim}
\usepackage{subfigure}
\usepackage{algorithm}
\usepackage{algpseudocode}
\usepackage{amsmath}
\usepackage{amssymb}
\newcommand{\argmin}[1]{\underset{#1}{\mathrm{argmin}}}

\usepackage{color}
\usepackage{varwidth}
\usepackage[normalem]{ulem}

\usepackage[table]{xcolor}

\usepackage{graphicx}
\usepackage{epstopdf}
\usepackage{float}
\usepackage{subfig}
\usepackage{epsfig}
\usepackage{fullpage}
\newcommand{\supp}{\operatorname{supp}} 

\newtheorem{theorem}{Theorem}[section]

\newtheorem{lemma}[theorem]{Lemma}
\newtheorem{proposition}[theorem]{Proposition}

\theoremstyle{definition}
\newtheorem{definition}[theorem]{Definition}

\newcommand{\prox}{{\rm prox}}
\newcommand{\rprox}{{\rm rprox}}

\theoremstyle{remark}
\newtheorem{remark}[theorem]{Remark}

\numberwithin{equation}{section}

\usepackage[T1]{fontenc}
\usepackage[utf8]{inputenc}
\usepackage[font=small,labelfont=bf]{caption}

\begin{document}

\title{Recovery guarantees for polynomial approximation from dependent data with outliers}

\author{Lam Si Tung Ho} 
\address{Department of Mathematics and Statistics, Dalhousie University, Halifax, NS, Canada}
\email{Lam.Ho@dal.ca}

\author {Hayden Schaeffer}
\address{Department of Mathematical Sciences, Carngie Mellon University, Pittsburgh, PA, USA}
 \email{schaeffer@cmu.edu}
 
\author{ Giang Tran}
\address{Department of Applied Mathematics, University of Waterloo, Waterloo, ON, Canada} 
  \email{giang.tran@uwaterloo.ca}
  
\author{ Rachel Ward}
\address{Department of Mathematics, The University of Texas at Austin, Austin, TX, USA}
   \email{rward@math.utexas.edu}

\begin{abstract}
  Learning non-linear systems from noisy, limited, and/or dependent data is an important task across various scientific fields including statistics, engineering, computer science, mathematics, and many more. In general, this learning task is ill-posed; however, additional information about the data's structure or on the behavior of the unknown function can make the task well-posed. In this work, we study the problem of learning nonlinear functions from corrupted and dependent data. The learning problem is recast as a sparse robust linear regression problem where we incorporate both the unknown coefficients and the corruptions in a basis pursuit framework. The main contribution of our paper is to provide a reconstruction guarantee for the associated $\ell_1$-optimization problem where the sampling matrix is formed from dependent data. Specifically, we prove that the sampling matrix satisfies the null space property and the stable null space property, provided that the data is compact and satisfies a suitable concentration inequality. We show that our recovery results are applicable to various types of dependent data such as exponentially strongly $\alpha$-mixing data, geometrically $\mathcal{C}$-mixing data, and uniformly ergodic Markov chain. Our theoretical results are verified via several numerical simulations.
\end{abstract}
% REQUIRED
%\begin{keywords}
%Recovery Guarantee,
%Basis Pursuit,
%Sparsity,
%Concentration Inequalities,
%Function Approximation,
%Dependent Data.
%\end{keywords}

% REQUIRED
%\begin{AMS}	
%68T05, 	%Learning and adaptive systems 
%  41A10, %	Approximation by polynomials
%   60F05, %Central limit and other weak theorems
%   	68Q32,  	%Computational learning theory 
% 62G08,  %	Nonparametric regression
% 	94A15, %  	Information theory, general 
%   	65K10.  	%Optimization and variational techniques
%\end{AMS}
\maketitle

\section{Introduction}
In the past few decades, there has been a rapid growth of interest in automated learning from data across various scientific fields including statistics \cite{vapnik1999overview}, engineering \cite{ljung1998system}, computer science \cite{jordan2015machine,roweis2000nonlinear}, mathematics, and many more. An overview of machine learning problems in a wide range of contexts (statistical learning theory, pattern recognition, system identification, deep learning, and so on) can be found in \cite{friedman2001elements, cherkassky2007learning, abu2012learning, goodfellow2016deep}. One of the main paradigms is to learn an unknown target function from a given collection of input-output pairs (supervised learning), which can be rephrased as the problem of finding an approximation of a multi-dimensional function. For example, in \cite{poggio1990regularization,poggio1990networks}, the authors demonstrated a connection between approximation theory and regularization with feedforward multilayer networks. In general, learning a smooth function from data is ill-posed unless a priori information about either the data structure or the generating function is provided \cite{tikhonov1977methods,  novak2009approximation, fornasier2012learning}.

One of the well-known methods to make the learning problem well-posed is to exploit additional properties of the target function \cite{friedman1994overview}. For example, if the target function depends only on a few active directions associated with a suitable random matrix, the function can be recovered from a small number of samples \cite{fornasier2012learning}. On the other hand, many well-known learning methods consider the target function in a particular function class (such as radial basis functions, projection pursuit, feed-forward neural networks, and tensor product methods)  and add a penalty (such as Tikhonov regularization or sparse constraints) to the associated parameter estimation problem.

Recently, sparse models combined with data-driven methods have been investigated intensively for learning nonlinear partial differential equations, nonlinear dynamical systems, and graph-based networks. The model selection problem for dynamical systems from time series dates back to \cite{crutchfield1987equation} where the authors investigate the concepts from dynamical system theory to recover the underlying structure from data. In \cite{yao2007modeling}, the authors construct a sampling matrix from the data matrix and its power to recover the ordinary differential equations  and find an optimal Kronecker product representation for the governing equations. Furthermore, based on the observation that many governing equations have a sparse representation with respect to high-dimensional polynomial spaces, the authors in \cite{brunton2016discovering} developed the SINDy algorithm which uses that sampling matrix and a sequential least-square thresholding algorithm to recover the governing equations of some unknown dynamical systems. The convergence of the SINDy algorithm is provided in \cite{zhang2018convergence}. 
 A group-sparse model was proposed in \cite{schaeffer2017learning} to learn governing equations from a family of dynamical systems with bifurcation parameters. By exploiting the cyclic structure of many nonlinear differential equations, the authors in \cite{schaeffer2018extracting} proposed an approach to identify the active basis terms using fewer random samples (in some cases on the order of a few snapshots). For the noisy case, in \cite{schaeffer2017sparse} the authors use the integral formulation of the differential equation to reduce the effect of noise and identify the model from a smoother basis set. To learn a nonlinear partial differential equation from spatio-temporal  dataset, the authors in \cite{schaeffer2017learningPDE} proposed a LASSO-based approach using a dictionary of partial derivatives. In \cite{rudy2017data}, the authors developed an adaptive ridge-regression version of \cite{brunton2016discovering} for learning nonlinear PDE, while in \cite{raissi2018hidden} a hidden physics model based on Gaussian processes was presented. On the other hand, the data are often contaminated by noise, contain outliers, have missing values, or have a limited amount of samples. When the given data are limited, there are several works addressing learning problems ranging from sampling strategies in high-dimensional dynamics using random initial conditions \cite{schaeffer2017extracting}, to a weighted $\ell_1$-minimization on the lower set \cite{rwl1weighted,chkifa2018polynomial},  model predictive control using SINDy \cite{kaiser2017sparse}, and sample complexity reduction to linear time-invariant systems \cite{fattahi2018data}. In \cite{shin2018sequential}, the authors proposed a method to approximate an unknown function from noise measurements via sequential approximation. Geometric methods, such as \cite{liao2016learning}, can be used to approximate functions in high-dimensions when the data concentrate on lower-dimensional sets.

Regarding supervised learning analysis, the input data are assumed to be independent and identically distributed (i.i.d.). However, this assumption does not hold in many applications such as speech recognition, medical diagnosis, signal processing, computational biology, and financial prediction. Alternatively, for non-i.i.d. processes satisfying certain mixing conditions, various reconstruction results have been addressed in different contexts.
The convergence rates of several machine learning algorithms have been studied for non-i.i.d. data. 
Examples include weighted average algorithm \cite{cuong2013generalization}, least squares support vector machines (LS-SVMs) \cite{hang2014fast}, and one-vs-all multiclass plug-in classifiers \cite{dinh2015learning}. In \cite{white1984nonlinear}, the authors discussed several mixing conditions for weakly dependent observations which guarantee the consistency and asymptotic normality for the nonlinear least squares estimator. Minimum complexity regression estimators for $m$-dependent observations and strongly mixing observations were proposed in \cite{modha1996minimum} using certain Bernstein-type inequalities for dependent observations. In \cite{ryabko2006pattern}, a conditionally i.i.d. model for pattern recognition was proposed, where the inputs are conditionally independent given the output labels. In \cite{steinwart2009learning}, the authors proved that if the data-generating process satisfies a certain law of large number, the support vector machines are consistent. In \cite{hang2017bernstein}, a Berstein-type inequality for geometrically $\mathcal{C}$-mixing processes is established and applied to deduce an oracle inequality for generic regularized empirical risk minimization algorithms. Using a strong central limit theorem for chaotic data and compressed sensing results, the authors in \cite{tran2017exact} proved a reconstruction guarantee for sparse reconstruction of governing equations for three-dimensional chaotic systems with outliers. The common technique in the mentioned works is the application of either a central limit theorem or a suitable concentration inequality for the given data.

In this work, we study the problem of learning nonlinear functions from identically distributed (but not necessarily independent) data that are corrupted by outliers and/or contaminated by noise. By expressing the target function in the multivariate polynomial space, the learning problem is recast as a sparse robust linear regression problem where we incorporate both the unknown coefficients and the corruptions in a basis pursuit framework.  The main contribution of our paper is to provide a reconstruction guarantee for the associated $\ell_1$-optimization problem where the (augmented) sampling matrix is formed from the data matrix, its powers, and the identity matrix. Although the data may not be i.i.d., we prove that the sampling matrix satisfies the null space property, provided that the data are compact and satisfies a suitable concentration inequality. Consequently, the basis pursuit problem will be guaranteed to have a unique solution and be stable with respect to noise. 
 
 The paper is organized as follows. In Section 2, we explain the problem setting. In Section 3, we first recall the theory from compressive sensing, then present the theoretical reconstruction guarantees. In Section 4, we state the recovery results for various types of data including i.i.d. data, exponentially strongly $\alpha$-mixing data, geometrically $\mathcal{C}$-mixing data, and uniformly ergodic Markov chain. The numerical implementations and results are described in Sections 5. We discuss the conclusion and future works in Section 6. 

\section{Problem Statement}
\label{sec:problemStatement}

We would like to learn a function $f:\mathbb{R}^d\rightarrow \mathbb{R},$
from $m$ data points $\Big(\mathbf{u}^{(i)}= \mathbf{x}^{(i)}+ \pmb{\theta}^{(i)}, y^{(i)} =f(\mathbf{x}^{(i)}) {+\varepsilon^{(i)}}\Big)_{i=1}^m$, where $\{\mathbf{u}^{(i)}\}$ is corrupted data, $\{\mathbf{x}^{(i)}\}$ is the uncorrupted part, $\{\pmb{\theta}^{(i)}\}$ represents the corruption, and $\{\varepsilon^{(i)}\}$ denotes noise. We say that $\mathbf{u}^{(i)}$ is an outlier if the corruption $\pmb{\theta}^{(i)}$ is non-zero. 
Assume that the function of interest $f$ is a multivariate polynomial of degree at most $p$:
\[f(x_1,\cdots,x_d) = \sum\limits_{|\alpha| = \alpha_1 +\ldots + \alpha_d\leq p}c^{\alpha}x_1^{\alpha_1}\ldots x_d^{\alpha_d}.\]
Let $\mathbf{y} = [y^{(1)}, \ldots, y^{(m)}]^T$, $\pmb{\varepsilon} =  [\varepsilon^{(1)}, \ldots, \varepsilon^{(m)}]^T$, $\pmb\theta$ be the matrix where the rows are $\pmb\theta^{(i)}$, and $U$ be the data matrix,
\begin{equation*}
   U 
   =  \begin{bmatrix}
       - & \mathbf{u}^{(1)} & -      \\
  - & \mathbf{u}^{(2)} & -   \\
  &\cdots&\\
  - & \mathbf{u}^{(m)} & -    \\
   \end{bmatrix}
 = \begin{bmatrix}
           u_1^{(1)} & \dots & u_d^{(1)}\\
           u_1^{(2)} & \dots & u_d^{(2)}\\
           \vdots  & \vdots & \vdots \\
            u_1^{(m)} & \dots & u_d^{(m)}
         \end{bmatrix}= \begin{bmatrix}| &  & | \\
       U_1 &\cdots & U_d \\
       | &  & | \end{bmatrix} _{m\times d}. 
       \end{equation*}
      
 \vspace{0.1 cm}
 \noindent Then we form the dictionary matrix $\Phi= \Phi_U$ from data,
 \begin{equation}
 \begin{aligned}
 \Phi_U = \begin{bmatrix}
 |  &    |     &    |      &         &    |      &    |       &      |        &           &     |              &\\
 1 & U_1  & U_2   & \dots  &   U_d & U_1^2 & U_1U_2 & \cdots & U_d^2  &\cdots \\
 |  &    |     &     |     &          &   |      &    |       &      |        &           &     |              &\\
 \end{bmatrix}_{m \times r}, 
 \end{aligned}
 \label{eqn:dictionary}
 \end{equation}
where $r =  {p + d\choose d}$ is the maximal number of $d$-multivariate monomials of degree at most $p$.

Denote $\mathbf{c} = (c^\alpha)_{|\alpha|\leq p}$ the coefficient vector and $\mathbf{e} = \mathbf{y} - \Phi \mathbf{c} $, we can reformulate our problem as follows: \\
\centerline{Find $(\mathbf{c},\mathbf{e}) \in\mathbb{R}^r\times \mathbb{R}^m$ such that $\mathbf{y} = \Phi \mathbf{c} + \mathbf{e}$.}\\
Without corruptions and with arbitrary noise vector $\pmb{\varepsilon}$, the problem is classically solvable by least squares regression once $m \geq r$.   With corruptions, whose locations can be arbitrary but are unknown beforehand, if $m \geq r$ and at least $n = r$ of the $m$ measurements are uncorrupted, then one could in theory do a regression on each of the ${m \choose n}$ subsets of $n$ measurements and retain the set with the smallest error; however, this is an infeasible combinatorial algorithm. Thus, the convex relaxation of this combinatorial algorithm is a natural choice for reconstruction algorithm:
\begin{equation}
\min_{\mathbf{c'}, \mathbf{e'}} \|\mathbf{e}'\|_1 \quad\text{subject to}\quad  \mathbf{y} = \Phi \, \mathbf{c}' + \mathbf{e}'.
\label{eqn:main0}
\end{equation}

On the other hand, if the polynomial coefficients are sparse or the polynomial function can be approximated by a sparse polynomial, the learning problem can be recast as follows:
\begin{equation}
\min_{\mathbf{c}',\mathbf{e}'} \|\mathbf{c}'\|_1 +\|\mathbf{e}'\|_1 \quad\text{subject to}\quad  \mathbf{y} = \Phi \, \mathbf{c}' + \mathbf{e}' ,
\label{eqn:main1}
\end{equation}
or, more generally, as the corrupted sensing problem \cite{li2013compressed, foygel2014corrupted, price2018},
\begin{equation}
\min_{\mathbf{c}',\mathbf{e}'} \|\mathbf{c}'\|_1 +\lambda \|\mathbf{e}'\|_1 \quad\text{subject to}\quad \mathbf{y} = \Phi \,\mathbf{c}' +\mathbf{ e}' .
\label{eqn:corruptedCS}
\end{equation}

For the remainder of the paper, we denote the sparsity level of $\mathbf{c}$ by $s_c$, and the row-sparsity level of $\pmb\theta$ by $s_{\theta}$. In the noiseless case ($\pmb\varepsilon =0$), we have:
\[\|\mathbf{e}\|_0 =\#\{i :  \pmb{\theta}^{(i)} \not = 0 \}\leq s_{\theta} .\] 

\section{Reconstruction Guarantee Analysis}
\label{sec:analysis}
Before presenting the properties of the matrix $[Id_m, \Phi]$ and theoretical guarantees for the corresponding $\ell_1$-optimization problems, we first recall some results from compressive sensing including the null space property and the stable null space property (see \cite{foucart2013mathematical} for a comprehensive overview). 
\subsection{Theory from Compressive Sensing}

\begin{definition}
A matrix $A\in \mathbb{R}^{m\times N}$ is said to satisfy 
\begin{itemize}
\item the null space property of order $s$ if
\[\|v_S\|_1 <\dfrac{1}{2}\|v\|_1\quad\text{for all }\ v\in\ker A\backslash \{0\},\]
for any set $S\subset [N] :=\{1,2,\ldots, N\}$ with $card(S)\leq s$.
\item the stable null space property of order $s$ with constant $0<\rho<1$ if 
\[\|v_S\|_1 \leq \dfrac{\rho}{\rho +1 }\|v\|_1\quad\text{for all }\ v\in\ker A,\]
for any set $S\subset [N]$ with $card(S)\leq s$.
\end{itemize}
\end{definition}
\begin{proposition}[Recovery guarantee given null space property]
Given a matrix $A\in \mathbb{R}^{m\times N}$, every $s$-sparse vector $z^*\in\mathbb{R}^N$ with $y = Az^*$ is the unique solution of 
\begin{equation}
\min \|z\|_1,\quad\text{subject to}\quad y = Az,
\label{eqn:l1min} 
\end{equation}
if and only if $A$ satisfies the null space property of order s.
\label{prop:l1min}
\end{proposition}
\begin{proposition}[Recovery guarantee given stable null space property]
Suppose a matrix $A\in \mathbb{R}^{m\times N}$ satisfies the stable null space property of order $s$ with constant $0<\rho<1$. Then, for any $x\in\mathbb{R}^N$ with $y=Ax$, a solution $z^\#$ of the optimization problem \eqref{eqn:l1min} approximates the vector $x$ with $\ell_1$-error 
\[\|x-z^\#\|_1 \leq \dfrac{2(1+\rho)}{1-\rho} \inf\limits_{\|w\|_0\leq s} \|x-w\|_1.\]
\label{prop:l1min_stable}
\end{proposition}
The null space property for the matrix $A$, along with the existence of an $s$-sparse solution to the underdetermined system of equations, is a sufficient and necessary condition for sparse solutions of the NP hard minimization problem,
\[\min \|z\|_0,\quad\text{subject to}\quad y = Az, \]
to be exactly recovered via the $\ell_1$-minimization \eqref{eqn:l1min}. On the other hand, the stable null space property of the matrix $A$ guarantees that any solution, sparse or not, can be recovered up to the error governed by its distance to $s$-sparse vectors. 

\subsection{Theoretical Guarantees}
We will show that if the uncorrupted data $\{\mathbf{x}^{(i)}\}$ satisfy an appropriate concentration inequality and their common distribution $\mu$ is \textit{non-degenerate} (that is, if $\mu(B) = 1$ implies $B$ contains infinitely many elements), then the polynomial coefficients of the unknown function as well as the location of the outliers can be exactly recovered with high probability from the unique solution of the $\ell_1$-minimization problem \eqref{eqn:main1}, provided that the output values $y$ are exact.  When the output values $y$ contain dense noise, we show that every solution of the associated optimization problem can be approximated by a sparse solution. 

To begin with, we will show that the matrix $[Id_m, \Phi_{m\times r}]$, where $\Phi_{m\times r} = \Phi_U$ is constructed from all monomials up to degree $p$, satisfies the null space property. 

\begin{theorem}
\label{thm:nspA}
Fix $p\in\mathbb{N}$. Consider $\{\mathbf{u}^{(i)} =\mathbf{x}^{(i)}  +\pmb{\theta}^{(i)}  \}_{i=1}^m\subset \mathbb{R}^d$ where the uncorrupted data  $\{\mathbf{x}^{(i)}\}$ are $L^{\infty}$-bounded by $B_\mathcal{X}$ and identically distributed according to a non-degenerate probability distribution $\mu$, and the corruption $\{\pmb{\theta}^{(i)}\}$ is $L^\infty$-bounded by $B_\Theta$ and $s_{\theta}$-row sparse. Let $r =  {p + d\choose d}$, and let $\kappa(\zeta, m)$ be a function such that \begin{equation}
\kappa(m^{-\delta}, m) \geq 3\delta  r \log m,
\label{eqn:kappa}
\end{equation}
when $m$ is large enough and $\delta >0$ is some chosen constant. 
Assume that $\{\mathbf{x}^{(i)}\}$ satisfies the following concentration inequality:
\begin{equation}
\Pr \left ( \left | \frac{1}{m}\sum_{i=1}^m{\varphi(\mathbf{x}^{(i)})} - \mathbb{E} [\varphi(\mathbf{x})] \right | \geq \zeta \right ) \leq  e^{ - \kappa(\zeta, m)},
\label{eqn:concentration}
\end{equation}
for any $\zeta > 0$ and any bounded Borel function $\varphi$. Here $e = \exp(1)$.

\indent Then, there is a constant $D>0$ depending only on $p,d$, and $\mu$, so that when $m$ satisifies:
\begin{equation}
\begin{aligned}
m &\geq \left(\max\{3+ 3B_{\mathcal{X}}^p, 4D^{-1} \} \right) ^{1/\delta}, \\
m&> \dfrac{4+8s\left(1+(B_{\mathcal{X}} +B_{\Theta})^p\right)}{D},
\label{eqn:boundm}
\end{aligned}
\end{equation}
the matrix $A= [Id_m, \Phi_{m\times r}]$, where $\Phi = \Phi_U$ is the dictionary matrix \eqref{eqn:dictionary}, satisfies the null space property of order $s\geq s_{\theta}$ with probability at least $(1 - m^{-\delta r})$.
\end{theorem}

\begin{proof}
	For each $c\in \mathcal{B} =\{v\in\mathbb{R}^r : \|v\|_1 =1\}$, define $\varphi^c:\mathbb{R}^d\rightarrow \mathbb{R}$ as follows:
\[
\varphi^c(\mathbf{x}) = \left | \sum\limits_{\alpha: |\alpha|\leq p} c^\alpha x_1^{\alpha_1}\ldots x_d^{\alpha_d} \right |.
\]
We first evaluate the lower bound for the summation $ \sum\limits_{i=1}^m{\varphi^c(\mathbf{x}^{(i)})}$.
For any non-zero $c\in\mathbb{R}^r$, we have $\mathbb{E}[\varphi^c(\mathbf{x})]~>~0$. 
Indeed, if $\mathbb{E}[\varphi^c(\mathbf{x})] =0$, then $\varphi^c(\mathbf{x}) = 0$ $\mu$-almost surely.
Since $\mu$ is non-degenerate, there are infinitely many $\mathbf{x}$ such that $\varphi^c(\mathbf{x}) = 0$.
This implies $c = 0$ which is a contradiction. Therefore, $\mathbb{E}[\varphi^c(\mathbf{x})] >0$  for any $c\in\mathcal{B}$.

On the other hand, since the set $\mathcal{B}$ is compact and nonempty, we can apply the extreme value theorem for the continuous function $\mathbb{E}[\varphi^c(\mathbf{x})]$ to get the following bound:
\[\inf_{c\in \mathcal{B}} \mathbb{E}[\varphi^c(\mathbf{x})]  \geq D > 0,\]
for some constant $D>0$. Note that $D$ depends on $\mu, d$, and $p$.

According to a well-known result on the covering number (for example, see Appendix C.2, \cite{foucart2013mathematical}), there exists a finite set of points $\mathcal{Q}$ in $\mathcal{B}$  
of cardinality \[|\mathcal{Q}|\leq   \left(3\, m^{\delta}(1+B_{\mathcal{X}}^p)\right)^r\] such that
	\begin{equation*}
\max\limits_{c\in\mathcal{B}} \min\limits_{q\in\mathcal{Q}} \|c-q\|_1\leq  \dfrac{1}{m^{\delta} (1+ B_{\mathcal{X}}^p )}.\end{equation*}

Applying the union bound on $\mathcal{Q}$ and using the assumption $\kappa(m^{-\delta},m)\geq 3\delta r\log m$, we derive:
\begin{equation*}
\Pr \Big ( \bigcup_{q \in \mathcal{Q}} \Big \{ \Big | \dfrac{1}{m}\sum_{i=1}^m{\varphi^q(\mathbf{x}^{(i)})} -  \mathbb{E} [\varphi^q(\mathbf{x})] \Big | \geq {m^{-\delta}} \Big \} \Big )  
\leq m^{\delta r}(3+3B_{\mathcal{X}}^p)^r e^{- \kappa(m^{-\delta}, m)  }\leq m^{-\delta r},
\end{equation*}
provided that 
	\begin{equation}
	m\geq \left(3+3B_{\mathcal{X}}^p\right)^{1/\delta}.\label{eqn:boundm00}
	\end{equation}
Hence,
\begin{equation*}
\Pr \left ( \max_{q \in \mathcal{Q}} \left \{ \left | \dfrac{1}{m}\sum_{i=1}^m{\varphi^q(\mathbf{x}^{(i)})} - \mathbb{E}[\varphi^q(\mathbf{x})] \right | \leq  m^{-\delta}\right \} \right ) 
\geq 1 - m^ {-\delta r}.%
\end{equation*}
Therefore, for any $q\in\mathcal{Q}$, we have:
\begin{equation}
\sum\limits_{i=1}^m  \varphi^q\left(\mathbf{x}^{(i)}\right)\geq m\mathbb{E}(\varphi^q(x)) - \left|  \sum_{i=1}^m{\varphi^q(\mathbf{x}^{(i)})} - m \mathbb{E}[\varphi^q(\mathbf{x})] \right | \geq mD - m^{1-\delta},\label{eqn:bound11}
\end{equation}
with probability at least $(1 - m^{-\delta r})$. 

For each $c\in\mathcal{B}$ there exists $q_c\in\mathcal{Q}$ so that $\|c-q_c\|_1~\leq~\dfrac{1}{m^{\delta}(1+B_{\mathcal{X}}^p)}$. Applying the H\"{o}lder's inequality for $\mathbf{x} = (x_j)\in \mathbb{R}^d$ with $\|\mathbf{x}\|_\infty\leq B_{\mathcal{X}}$, we obtain:
\[\left | \sum\limits_{\alpha: |\alpha|\leq p}(c^\alpha - q_c^\alpha) \prod_{j=1}^dx_j^{\alpha_j}\right | \leq \|c-q_c\|_1 \max\limits_{\alpha: |\alpha|\leq p}\prod_{j=1}^d\left | x_j^{\alpha_j}\right | \leq m^{-\delta}.\]
Combining with the inequality \eqref{eqn:bound11}, we obtain
\begin{equation*}
\begin{aligned}
\sum\limits_{i=1}^m \varphi^c\left(\mathbf{x}^{(i)}\right) &\geq \sum\limits_{i=1}^m  \varphi^{q_c}\left(\mathbf{x}^{(i)}\right) -\sum\limits_{i=1}^m  \left | \sum\limits_{\alpha: |\alpha|\leq p}(c^\alpha - q_c^\alpha) \prod_{j=1}^d\left(x_j^{(i)}\right)^{\alpha_j}\right |\\
&\geq m(D - 2m^{-\delta})\\
&\geq \dfrac{1}{2}mD,
\end{aligned}
\end{equation*}
with probability at least $(1 - m^{-\delta r})$,
provided that 
\begin{equation}
m\geq \left(\frac{4}{D}\right)^{1/\delta}.
\label{eqn:boundm01}
\end{equation}

By linearity, we have in the same event, 
\begin{equation} 
\sum_{i=1}^m{\varphi^c(\mathbf{x}^{(i)})} \geq \frac{1}{2}m D \, \|c\|_1, \quad \forall c \in\mathbb{R}^r\setminus\{0\}.  
\label{eqn:lowerbound00}
\end{equation}
Next, we will estimate the lower bound for $\|\Phi c\|_1$, where $c\in \mathbb{R}^r \setminus \{0\} $. Denote $R =(R_1,\ldots, R_m)^T \in \mathbb{R}^m $, where $R_i$ is defined as follows
\[R_i = (\Phi c)_i -  \sum\limits_{\alpha: |\alpha|\leq p} c^\alpha \left(x_1^{(i)}\right)^{\alpha_1}\ldots \left(x_d^{(i)}\right)^{\alpha_d}, \quad i=1,\ldots, m.\]
Applying the H\"{o}lder's inequality, we have
\begin{equation*}
\begin{aligned}
| (\Phi c)_i |&=\Big |\sum\limits_{\alpha: |\alpha|\leq p} c^\alpha \left(u_1^{(i)}\right)^{\alpha_1}\ldots \left(u_d^{(i)}\right)^{\alpha_d}\Big |\leq \sum\limits_{\alpha: |\alpha|\leq p} \left |c^\alpha   \left(u_1^{(i)}\right)^{\alpha_1}\ldots \left(u_d^{(i)}\right)^{\alpha_d}\right| \\
&\leq \|c\|_1\max\limits_{\alpha: |\alpha|\leq p}\prod\limits_{j=1}^d\left|\left(u_j^{(i)}\right)^{\alpha_j}\right|\leq \|c\|_1 \left(1+(B_{\mathcal{X}} + B_{\Theta})^p\right) .
\label{eqn:boundPhi}
\end{aligned}
\end{equation*}
Similarly, we have \[\varphi^c(\mathbf{x}^{(i)})\leq \|c\|_1 \left(1+B_{\mathcal{X}}^p\right).\] 
Therefore,
\[
\left | R_i\right | \leq   | (\Phi c)_i | + \varphi^c\big(\mathbf{x}^{(i)}\big) \leq  2\|c\|_1 (1 + (B_{\mathcal{X}}+B_\Theta)^p ).
\]
Since $\|\Theta\|_{2,0}\leq s_{\theta}$, we deduce $\|R\|_0\leq s_{\theta}$ and 
\begin{equation}
\|R\|_1 =\sum\limits_{i=1}^m |R_i| \leq 2 s_{\theta} \|c\|_1 (1 + (B_{\mathcal{X}}+B_\Theta)^p ).
\label{eqn:boundR}
\end{equation}
Thus, in the event that \eqref{eqn:lowerbound00} holds, we have combined with \eqref{eqn:boundR} that 
\begin{equation}
\begin{aligned}
\| \Phi c \|_1 & \geq \sum\limits_{i=1}^m\varphi^c\left(\mathbf{x}^{(i)}\right)- \| R \|_1\\
& \geq \|c\|_1 \left( \frac{1}{2}m D   - 2{s_{\theta}}(1 + (B_{\mathcal{X}}+B_\Theta)^p )\right)\\
&\geq \dfrac{1}{4} m D\|c\|_1,
\end{aligned}
\label{eqn:boundPhi2}
\end{equation}
provided moreover that
\begin{equation*}
m \geq \dfrac{8s_{\theta}(1 + (B_{\mathcal{X}}+B_\Theta)^p )}{D} =\widetilde{C}.
\label{bound_m2}
\end{equation*}
Now, we are ready to verify the null space property condition for $A = [Id_m,\Phi_{m\times r}]$ in the event that \eqref{eqn:lowerbound00} holds.  
Let $S\subset [m+r]$ be an arbitrary set of size $s$ and $w\in\ker A\setminus \{\vec{0}\}$. Denote $\hat{c}\in\mathbb{R}^r $ be the last $r$ entries of $w$, and 
\[S_1 =S\cap [m], \quad S_2= \left(S\cap\{m+1,\ldots, m+r\}\right) - m\subset [r].\]
Since $w\in\ker A \setminus \{\vec{0}\}$, $\hat{c}\not = \vec{0}_r$ and  $w =[-\Phi \hat{c},\hat{c}]$. Using the inequality \eqref{eqn:boundPhi}, we have
\[
\|w_S\|_1  = \|\hat{c}_{S_2}\|_1 + \|(\Phi \hat{c})_{S_1}\|_1   \leq  \|\hat{c}\|_1+  \|(\Phi \hat{c})_{S_1}\|_1  \leq  \|\hat{c} \|_1  \left( 1+ s\left(1+(B_{\mathcal{X}} +B_\Theta)^p\right) \right).\]
On the other hand, using the inequality \eqref{eqn:boundPhi2}, we obtain
\[
 \|w\|_1 = \|\hat{c}\|_1 + \|\Phi \hat{c}\|_1 \geq  \|\hat{c}\|_1\left(1+ \frac{1}{4}mD\right).
\]
Then when $m$ satisfies \eqref{eqn:boundm00}, \eqref{eqn:boundm01}, and
\begin{equation}
m> \dfrac{4+ 8s(1+(B_{\mathcal{X}}+B_\Theta)^p)}{D}>\widetilde{C},
\label{eqn:mNSP} 
\end{equation}
we have $\|w_S\|_1 < \dfrac{1}{2}\|w\|_1$, for any $w\in\ker A \setminus \{\vec{0}\}$. That completes our proof.
\end{proof}
\begin{remark}\label{rm:stableNSP} 
	\begin{itemize}
		\item Since $\|\Phi c\|_1\geq \frac{1}{4}mD\|c\|_1$ for any $c\in\mathbb{R}^r\setminus \{0\}$ with probability $1-m^{-\delta r}$, we conclude that the matrix $\Phi$ is of full column rank.
		\item From the proof, we also derive that if $s\geq r$, the matrix $[Id_m,\Phi_{m\times r}]$ satisifes the partial null space property of order $s-r$ (see \cite{bandeira2013partial}, Definition 3.1). 
		\item 	If we keep the conditions \eqref{eqn:boundm00} and \eqref{eqn:boundm01}, and change the condition \eqref{eqn:mNSP} to
	\begin{equation}
	m> \dfrac{4+ 4s(\rho+1)(1+(B_{\mathcal{X}}+B_\Theta)^p)}{\rho D},
	\label{eqn:mSNSP} 
	\end{equation}
	then $\|w_S\|_1 \leq \dfrac{\rho}{\rho+1}\|w\|_1$, for any $w\in\ker A$ and any set $S\subset [m+r]$ with $card(S)\leq s$. It means $A$ satisfies the stable null space property of order $s$. 
	\end{itemize}
	
	\end{remark}

\medskip

Combining with the reconstruction results from compressed sensing (see Proposition \ref{prop:l1min} and Proposition \ref{prop:l1min_stable}), we immediately obtain the following reconstruction guarantees.
\begin{theorem}Fix $p\in\mathbb{N}$.
Suppose we observe corrupted measurements
\[  \Big(\mathbf{u}^{(i)}= \mathbf{x}^{(i)}+ \pmb{\theta}^{(i)}, y^{(i)} =f(\mathbf{x}^{(i)})+\varepsilon^{(i)}\Big)_{i=1}^m \subset {\mathbb{R}^d}\times {\mathbb{R}},\]
where $\{\mathbf{x}^{(i)}\}$ and $\{\pmb{\theta}^{(i)}\}$ satisfy the assumptions in Theorem \ref{thm:nspA}, and $f$ is a sparse multivariate polynomial with at most $s_c$ monomial terms of degree at most $p$. Denote $\mathbf{y} = (y^{(i)})$, $s=s_c+s_{\theta}$, $\Phi_U$ be the dictionary matrix \eqref{eqn:dictionary}, and $\mathbf{c}$ be the unknown polynomial coefficients of $f$. The problem can be recast as
\[\mathbf{y} = \Phi \mathbf{c} +\mathbf{e},\]
for some $\mathbf{e}\in \mathbb{R}^m$. 
\begin{itemize}
\item[(a)] When $\pmb{\varepsilon}  =0$, then $\supp \mathbf{e} =  \{i: \pmb{\theta}^{(i)} \not = 0\}$. Suppose $\|\mathbf{c}\|_0 + \|\mathbf{e}\|_0\leq s$, then there is a constant $D>0$ depending only on $p,d$, and $\mu$, so that when $m$ satisfies \eqref{eqn:boundm}, the polynomial coefficients $\mathbf{c}$ of $f$ as well as the vector $\mathbf{e}$ can be exactly recovered with probability $(1 - m^{-\delta r})$ from the unique solution to the $\ell_1$-minimization problem: 
\[\min\limits_{\mathbf{c}',\mathbf{e}'}\|\mathbf{e}'\|_1+\|\mathbf{c}'\|_1 \quad\text{subject to}\ \Phi \mathbf{c}' +\mathbf{e}'  = \mathbf{y}.\]
\item[(b)] When $\pmb{\varepsilon} \not = 0$ and is not necessarily sparse, if $m$ satisfies \eqref{eqn:mSNSP}, \eqref{eqn:boundm00}, and \eqref{eqn:boundm01}, a solution  $(\mathbf{c}^\#, \mathbf{e}^\#)$ to the $\ell_1$-minimization \eqref{eqn:main1} approximates the true solution $(c,e)$ with $\ell_1$-error: 
\[ \|\mathbf{c} - \mathbf{c}^\#\|_1 + \|\mathbf{e} - \mathbf{e}^\# \|_1  \leq \dfrac{2(1+\rho)}{1- \rho }\left(\|\mathbf{c} -\mathbf{c}^*\|_1+ \|\mathbf{e} -\mathbf{e}^*\|_1\right),\]
where $ [\mathbf{c}^*,\mathbf{e}^*]$ is the best $s$-term approximation (vector of $s$ largest-magnitude entries)  of $[\mathbf{c},\mathbf{e}]$ and $\rho\in (0,1)$ is the stable null space constant of the matrix $[Id_m,\Phi_U]$.
\end{itemize}
\label{thm:main}
\end{theorem}

\begin{remark} 
\begin{itemize} 
\item The partial $\ell_1$-minimization problem in \cite{tran2017exact}
\[\min_{\mathbf{c}',\mathbf{e}'} \|\mathbf{e}'\|_1 \quad\text{subject to}\quad \mathbf{y} = \Phi\mathbf{c}'+\mathbf{e}',\]
is a special case of problem \eqref{eqn:main1} with $s=s_{\theta}\geq r$. In other words, given corrupted input-output data where the corruption measurements are $s_{\theta}$-sparse, we can recover the polynomial function that fit the given data and detect the outliers correctly.

\item The same result in Theorem \ref{thm:main} can be extended immediately to learn a system of high-dimensional polynomial functions $\mathbf{f} = (f_1,\ldots, f_n)\in\mathbb{R}^n$ with the same coefficient matrix, where each $f_j$ is a multivariate polynomial of degree at most $p$:
\[\min_{\mathbf{c}'_j,\mathbf{e}'_j} \|\mathbf{c}'_j\|_1 +\|\mathbf{e}'_j\|_1\quad\text{subject to}\quad \mathbf{y}_j = \Phi \mathbf{c}'_j + \mathbf{e}'_j, \quad 1\leq j\leq n.\]
\item  By considering a slight modification of the matrix $A$, $\widetilde{A} = \left [\dfrac{1}{\lambda} Id_m, \Phi\right ]$, we can verify that $\widetilde{A}$ also satisfies the null space property, provided that $m$ is sufficiently large. Indeed, every $w\in\ker \widetilde{A}\setminus\{0\}$ can be written as $w~=~ [-\lambda \Phi \hat{\mathbf{c}}, \hat{\mathbf{c}}]$. Then with the lower bound on $\|\Phi \mathbf{c}\|_1$, we can immediately show $\|w_S\|_1< \dfrac{1}{2} \|w\|_1$, provided that 

\begin{equation}
\begin{aligned}
m &\geq \left(\max\{3+ 3B_{\mathcal{X}}^p, 4D^{-1} \} \right) ^{1/\delta}, \\
m&> \dfrac{8s\left(1+(B_{\mathcal{X}} +B_{\Theta})^p\right)}{D},\\
\lambda&> \dfrac{4}{mD - 8s\left(1+(B_{\mathcal{X}} +B_{\Theta})^p\right)} .
\end{aligned}
\end{equation}
Hence, the corrupted compressed sensing problem
\[\min\limits_{\mathbf{c}', \mathbf{e}'} \|\mathbf{c}'\|_1 +\lambda \|\mathbf{e}'\|_1, \quad\text{subject to } \mathbf{y} = \Phi \mathbf{c}' +\mathbf{e}',\]
will have a unique solution.
\end{itemize}
\end{remark}
\section{Recovery Results for Various Types of Data}
In this section, we apply our results to several popular types of dependent data. Indeed, we only need to verify that these types of data satisfy the required concentration inequality in Theorem \ref{thm:nspA}. For the sake of simplicity, we state the recovery results for the noiseless case of $\mathbf{y}$ (i.e., when $\varepsilon^{(i)} =0$).

\subsection{Independent and Identically Distributed (i.i.d.) Data}

In \cite{shorack2009empirical}, the authors provide the following Bernstein inequality for i.i.d. random variables:

\begin{lemma}
If $\{\mathbf{x}^{(i)}\}$ are i.i.d. random variables with $|\varphi(\mathbf{x}^{(1)})-\mathbb{E}(\varphi(\mathbf{x}^{(1)}))| \leq C_1\ a.s.$, then the following probability inequality holds
\begin{equation}
\Pr \left ( \bigg| \frac{1}{m}\sum_{i=1}^m{\varphi(\mathbf{x}^{(i)})} -  \mathbb{E} [\varphi(X)] \bigg| \geq \zeta \right ) \leq 2 \exp \left ( -\frac{\zeta^2 m }{C_2 + C_3\zeta} \right ),
\label{eqn:bernstein_iid}
\end{equation}
where
 \[ C_2 = 2\mathbb{E}[\varphi^2(\mathbf{x}^{(1)})] - 2(\mathbb{E}[\varphi(\mathbf{x}^{(1)})])^2,\  C_3 =  \frac{2}{3}C_1, \] 
 and $\varphi$ is any bounded Borel function.
 \label{lem:Berstein_iid}
\end{lemma}

In this case, the function $\kappa$ in the concentration inequality \eqref{eqn:concentration} is 
\[\kappa(\zeta,m) = \dfrac{\zeta^2m}{C_2+C_3\zeta} -\log 2,\]
and satisfies the condition \eqref{eqn:kappa} for any constant $\delta\in \Big(0,\dfrac{1}{2}\Big)$, when $m$ is large enough. Indeed, the condition on $\kappa$ can be re-written as
\begin{equation}
r\leq \dfrac{1}{3\delta\log m}\left( \dfrac{m}{C_2m^{2\delta} + C_3m^{\delta}}-\log 2\right).\label{eqn:boundr_iid}\end{equation}
If the maximal polynomial degree $p$ is fixed, the smaller $\delta$ is, the smaller $m$ is needed to satisfy the inequality \eqref{eqn:boundr_iid}.

As a result, we have the following recovery result for i.i.d data.
\begin{theorem}
Fix $p\in\mathbb{N}$. Suppose we observe corrupted measurements
\[  \Big(\mathbf{u}^{(i)}= \mathbf{x}^{(i)}+ \pmb{\theta}^{(i)}, y^{(i)} =f\big(\mathbf{x}^{(i)}\big)\Big)_{i=1}^m \subset {\mathbb{R}^d}\times {\mathbb{R}},\]
where the uncorrupted data $\{\mathbf{x}^{(i)}\}$ are i.i.d. according to a non-degenerate distribution $\mu$ and $L^\infty$-bounded by $B_{\mathcal{X}}$; the corruption $\{\pmb{\theta}^{(i)}\}$ is $L^\infty$-bounded by $B_{\Theta}$ and $s_{\theta}$-row sparse; and $f$ is a sparse multivariate polynomial with at most $s_c$ monomials of degree at most $p$. Then, when $m$ satisfies \eqref{eqn:boundm} and \eqref{eqn:boundr_iid}, the polynomial coefficients of the function $f$ can be exactly recovered and the outliers can be successfully detected from the unique solution of \eqref{eqn:main1} with high probability.
\end{theorem}

\subsection{Exponentially Strongly $\alpha$-mixing Data}
We first recall the definition of $\alpha$-mixing coefficients and a concentration inequality for $\alpha$-mixing. 
For a stationary stochastic process $\{\mathbf{x}_t\}$, define (see \cite{rosenblatt1956central, modha1996minimum})
\[
\alpha (s) =\sup_{\substack{-\infty<t<\infty \\  A\in\sigma(\mathbf{x}^-_{t}), B\in\sigma (\mathbf{x}^+_{t+s})}} |\Pr(A\cap B) - \Pr(A)\Pr(B)|.
\]
The stochastic process is said to be {\it exponentially strongly $\alpha$-mixing} if 
\[\alpha(s) \leq \overline{\alpha}\exp(-c_{\alpha} s^\beta),\quad s\geq 1,\]
for some $\overline{\alpha}>0,\beta>0,$ and $c_\alpha>0$, where the constants $\beta$ and $c_\alpha$ are assumed to be known.
Note that strong mixing implies asymptotic independence over sufficiently large time. 

In  \cite{modha1996minimum}, the authors proved the following concentration inequality for exponentially strongly $\alpha$-mixing:

\begin{lemma}
If $\{\mathbf{x}^{(i)}\}$ are stationary exponentially strongly $\alpha$-mixing with \\
$|\varphi(\mathbf{x}^{(1)}) - \mathbb{E}(\varphi(\mathbf{x}^{(1)}))| \leq C_0 \, \text{a.s}.$, then the following probability inequality holds for $m$ sufficiently large:
\begin{equation}
\Pr \left ( \bigg | \frac{1}{m}\sum_{i=1}^m{\varphi(\mathbf{x}^{(i)})} - \mathbb{E} [\varphi(\mathbf{x}^{(1)})] \bigg | \geq \zeta \right ) \leq C_1 \exp \left ( -\frac{\zeta^2 m_\alpha }{C_2 + C_3 \zeta} \right ),
\label{eqn:bernstein_alpha}
\end{equation}
where
 \[m_\alpha := \left \lfloor \frac{m}{\lceil (8m/c_\alpha)^{1/(\beta + 1)} \rceil} \right \rfloor =C_\alpha\, m^{\beta/(\beta+1)},\]
 \[ C_1 = 2(1+4e^{-2}\overline{\alpha}), \ C_2 = 2\mathbb{E}(\varphi^2(\mathbf{x}^{(1)})) - 2(\mathbb{E}(\varphi(\mathbf{x}^{(1)})))^2,\  C_3 =  \frac{2}{3}C_0, \] and $\varphi$ is any bounded Borel function.
\end{lemma}

Hence the concentration inequality \eqref{eqn:concentration} is satisfied with 
\[\kappa(\zeta,m) =\dfrac{\zeta^2 m_\alpha }{C_2 + C_3 \zeta}  -\log C_1.\]
Since
\begin{equation}
\kappa(m^{-\delta},m) = \dfrac{m_{\alpha}}{C_2m^{2\delta} + C_3 m^\delta} -\log C_1\geq 3\delta r\log m,
\label{eqn:boundr_alpha}
\end{equation}
for any $\delta\in \bigg(0,\dfrac{\beta}{2(\beta+1)}\bigg)$ when $m$ is lare enough, we have the recovery result for exponentially strongly $\alpha$-mixing data.

\begin{theorem}
Fix $p\in\mathbb{N}$. Suppose we observe corrupted measurements
\[  \Big(\mathbf{u}^{(i)}= \mathbf{x}^{(i)}+ \pmb{\theta}^{(i)}, y^{(i)} =f(\mathbf{x}^{(i)})\Big)_{i=1}^m \subset {\mathbb{R}^d}\times {\mathbb{R}},\]
where the uncorrupted data $\{\mathbf{x}^{(i)}\}$ are stationary exponentially strongly $\alpha$-mixing and $L^\infty$-bounded by $B_{\mathcal{X}}$; the corruption $\{\pmb{\theta}^{(i)}\}$ is $L^\infty$-bounded by $B_{\Theta}$ and $s_{\theta}$-row sparse; and $f$ is a sparse multivariate polynomial with at most $s_c$ monomials of degree at most $p$. 
If the stationary distribution $\mu$ of $\{\mathbf{x}^{(i)}\}$ is non-degenerate, then when $m$ is sufficiently large and satisfies Equation \eqref{eqn:boundm} and Equation \eqref{eqn:boundr_alpha}, the polynomial coefficients of the function $f$ can be exactly recovered and the outliers can be successfully detected from the unique solution of \eqref{eqn:main1} with high probability.
\end{theorem}

\subsection{Geometrically (time-reversed) $\mathcal{C}$-mixing Data}
The $\mathcal{C}$-mixing processes were introduced in \cite{ maume2006exponential} to exhibit many common dynamical systems that are not necessary $\alpha$-mixing such as  Lasota-Yorke maps, uni-modal maps, piecewise expanding maps in higher dimension. Moreover, the geometrically $\mathcal{C}$-mixing processes are strongly related to some well-known results on the decay of correlations for dynamical systems (see \cite{hang2017bernstein}).

Let $\{\mathbf{x}^{(i)}\}$ be an $\mathcal{X}$-valued stationary process on $(\Omega, \mathcal{A}, \mu)$.
For a semi-norm $\|\cdot\|$ on a vector space of bounded measurable functions that satisfies $\| e^h \| \leq \|e^h\|_\infty \| h \|$, we define the $\mathcal{C}$-norm by $\| h \|_\mathcal{C} = \| h \|_\infty + \| h\|$.

Let $\mathcal{A}^i_1$ and $\mathcal{A}^\infty_{i+m}$ be the $\sigma$-algebras generated by $(\mathbf{x}^{(1)}, \ldots, \mathbf{x}^{(i)})$ and $(\mathbf{x}^{(i+m)}, \mathbf{x}^{(i+m+1)}, \ldots)$ respectively.
Then, the $\mathcal{C}$-mixing coefficient is
\begin{multline*}
\phi_{\mathcal{C}}(m) = \sup \Big\{ \Big | \mathbb{E}[Z h(\mathbf{x}^{(i+m)})] - \mathbb{E}(Z) \mathbb{E}[h(\mathbf{x}^{(i+m)})] \Big | : i \geq 1, \\
Z ~~ \text{is $\mathcal{A}^i_1$-measureable and} ~~ \|Z\|_1 \leq 1, \| h \|_{\mathcal{C}} \leq 1\Big\},
\end{multline*}
and the time-reversed $\mathcal{C}$-mixing coefficient is
\begin{multline*}
\phi_{\mathcal{C}, \text{rev}}(m) = \sup \Big\{ \Big | \mathbb{E}[Z h(X^{(i)})] - \mathbb{E}(Z) \mathbb{E}[h(X^{(i)})] \Big | : i \geq 1, \\
Z ~~ \text{is $\mathcal{A}^\infty_{i+m}$-measureable and} ~~ \|Z\|_1 \leq 1, \| h \|_{\mathcal{C}} \leq 1\Big\}.
\end{multline*}

A sequence of random variables $\{\mathbf{x}^{(i)}\}$ is called geometrically (time-reversed) $\mathcal{C}$-mixing if
\[
\phi_{\mathcal{C}, (\text{rev})}(m) \leq c \exp(-b m^\beta), \quad m \geq 1,
\]
for some constants $b>0, c \geq 0$, and $\beta > 0$.
The following concentration inequality for stationary geometrically (time-reversed) $\mathcal{C}$-mixing process is a direct consequence of the Bernstein inequality presented in \cite{hang2017bernstein}.

\begin{lemma}
Let $\{\mathbf{x}^{(i)}\}_{i\geq 1}$ be a stationary geometrically (time reversed) $\mathcal{C}$-mixing process. Consider a function $\varphi : \mathcal{X} \to \mathbb{R}$ such that $\| \varphi \| \leq A$, $\| \varphi \|_\infty \leq B$, and $\text{Var}(\varphi(\mathbf{x}^{(1)})) \leq \sigma^2$.
Then, for sufficient large $m$ we have 
\begin{equation}
\Pr \left ( \bigg | \frac{1}{m}\sum_{i=1}^m{\varphi(\mathbf{x}^{(i)})} - \mathbb{E} [\varphi(\mathbf{x}^{(1)})] \bigg | \geq \zeta \right ) \leq 4\exp \left ( -\frac{m\zeta^2 }{ 8 (\log m)^{2/\beta} \left(\sigma^2 +  \zeta B/3\right)} \right ).
\label{eqn:bernstein_C}
\end{equation}
\end{lemma}

In this case, the concentration inequality \eqref{eqn:concentration} holds for 
\[\kappa(\zeta,m) = \dfrac{m\zeta^2 }{ 8 (\log m)^{2/\beta} \left(\sigma^2 +  \zeta B/3\right)} -\log 4,\] 
and satisfies the condition \eqref{eqn:kappa} for any $\delta \in \Big(0,\dfrac{1}{2}\Big)$ when $m$ is large enough. Hence, we have the recovery result for geometrically (time-reversed) $\mathcal{C}$-mixing data.

\begin{theorem}
Fix $p\in\mathbb{N}$. Suppose we observe corrupted measurements
\[  \Big(\mathbf{u}^{(i)}= \mathbf{x}^{(i)}+ \pmb{\theta}^{(i)}, y^{(i)} =f(\mathbf{x}^{(i)})\Big)_{i=1}^m \subset {\mathbb{R}^d}\times {\mathbb{R}},\]
where the uncorrupted data $\{\mathbf{x}^{(i)}\}$ are stationary geometrically (time-reversed) $\mathcal{C}$-mixing with respect to the semi-norm $\| h \| = \sup\limits_{X \in \mathcal{X}} \| \nabla h(X) \|_1$ and $L^\infty$-bounded by $B_{\mathcal{X}}$; the corruption $\{\pmb{\theta}^{(i)}\}$ is $L^\infty$-bounded by $B_{\Theta}$ and $s_{\theta}$-row sparse; and $f$ is a sparse multivariate polynomial with at most $s_c$ monomials of degree at most $p$. 
If the stationary distribution $\mu$ of $\{\mathbf{x}^{(i)}\}$ is non-degenerate, then when $m$ is sufficiently large, the polynomial coefficients of the function $f$ can be exactly recovered and the outliers can be successfully detected from the unique solution of \eqref{eqn:main1} with high probability.
\end{theorem}

\subsection{Uniformly Ergodic Markov Chain}
Let $\{\mathbf{x}^{(i)}\}$ be a Markov chain on $(\Omega, \mathcal{A})$ with a unique stationary distribution $\mu$.
We define: 
\[
P^k(x, A) = \Pr(\mathbf{x}_{k+i} \in A \mid \mathbf{x}_i = x).
\]
The chain $\{\mathbf{x}^{(i)}\}$ is called \textit{uniformly ergodic} if
\[
\sup_x{\| P^k(x, .) - \mu(.)\|_{TV}} \to 0 \quad \text{as} ~~ k \to \infty.
\]
In this case, there exists a positive integer $k_0$, $\lambda > 0$, and a probability distribution $\rho$ such that
\[
P^{k_0}(x, A) \geq \lambda \rho(A), \quad \text{for all} ~~ x \in \mathcal{X}, A \in \mathcal{A}.
\]

We have the following concentration inequality for uniformly ergodic Markov chain \cite{kontoyiannis2005relative}:
\begin{lemma}
	Let $\{\mathbf{x}^{(i)}\}$ be a stationary uniformly ergodic Markov chain. Then for any $\varphi : \mathcal{X} \to \mathbb{R}$ such that $\| \varphi \|_\infty \leq B$, any $\zeta > 0$, and $m \geq 1 + 3k_0B/(\lambda \zeta)$, we have 
\begin{equation}
\Pr \left ( \bigg | \frac{1}{m}\sum_{i=1}^m{\varphi(\mathbf{x}^{(i)})} -  \mathbb{E} [\varphi(\mathbf{x}^{(1)})] \bigg | \geq \zeta \right ) \leq  2 \exp \left [ - \frac{m-1}{2} \left ( \frac{\lambda}{k_0 B} \zeta - \frac{3}{m-1} \right )^2 \right ].
\label{eqn:bernstein_UE}
\end{equation}
\end{lemma}

In this case, the concentration inequality \eqref{eqn:concentration} holds for 
\[\kappa(\zeta,m) = \dfrac{m-1}{2} \left ( \dfrac{\lambda}{k_0 B} \zeta - \frac{3}{m-1} \right )^2 -\log 2.\]
Observe that
\begin{align*}
	\kappa(m^{-\delta}, m) &= \frac{m-1}{2} \left ( \frac{\lambda}{k_0 B} m^{-\delta} - \frac{3}{m-1} \right )^2  - \log 2 \\
	&= \frac{\lambda^2 (m-1)m^{-2 \delta}}{2 k_0^2 B^2} - \frac{3 \lambda m^{-\delta}}{k_0 B} + \frac{9}{2(m-1)} - \log 2 \\
	&\geq  3\delta r \log m,
\end{align*}
for any $\delta \in \Big (0,\dfrac{1}{2}\Big)$ when $m$ is large enough. Therfore, we have the recovery result for uniformly ergodic Markov chain data.

\begin{theorem}
Fix $p\in\mathbb{N}$. Suppose we observe corrupted measurements
\[  \Big(\mathbf{u}^{(i)}= \mathbf{x}^{(i)}+ \pmb{\theta}^{(i)}, y^{(i)} =f(\mathbf{x}^{(i)})\Big)_{i=1}^m \subset {\mathbb{R}^d}\times {\mathbb{R}},\]
where the uncorrupted data $\{\mathbf{x}^{(i)}\}$ 
 form a stationary uniformly ergodic Markov chain and $L^\infty$-bounded by $B_{\mathcal{X}}$; the corruption $\{\pmb{\theta}^{(i)}\}$ is $L^\infty$-bounded by $B_{\Theta}$ and $s_{\theta}$-row sparse; and $f$ is a sparse multivariate polynomial with at most $s_c$ monomials of degree at most $p$. 
 If the stationary distribution $\mu$ of $\{\mathbf{x}^{(i)}\}$ is non-degenerate,
then when $m$ is sufficiently large, the polynomial coefficients of the function $f$ can be exactly recovered and the outliers can be successfully detected from the unique solution of \eqref{eqn:main1} with high probability.
\end{theorem}

%
%
%
%%&&&&&&&&&&&&&&&&&&&&&&&&&&&&&&&&&&&&&&&&&&&&&&&&&&&&&&&&&&&&&&&&&&&&&&&&\\
%%&&&&&&&&&&&&&&&&&&&&&&&&&&&&&&&&&&&&&&&&&&&&&&&&&&&&&&&&&&&&&&&&&&&&&&&&\\
%%&&&&&&&&&&&&&&&&&&&&&&&&&&&&&&&&&&&&&&&&&&&&&&&&&&&&&&&&&&&&&&&&&&&&&&&&\\
%
%
%
\section{Numerics and Computational Results}
\label{sec:numerics}
In this section, we verify the exact recovery of polynomial coefficients from data sampled from a stationary process with sparse random corruptions. For each of the examples, we use exponentially strong $\alpha$-mixing data, although the other related processes would produce similar results. Our computational tests verify the recovery results from Theorem \ref{thm:main} as well as the method's dependence on parameters such as the sampling rate, polynomial degree, the sampling distribution of the corruption vector, and the sparsity of the corruption vector. 

\subsection{Algorithm}
To solve the constrained optimization problem \eqref{eqn:corruptedCS}, we use the well-known Douglas-Rachford algorithm \cite{lions1979splitting, combettes2011proximal}. Equation \eqref{eqn:corruptedCS}
can be written as:
\begin{equation}
\min_{w} \|w\|_1 \quad\text{subject to}\quad  y = {A}w,
\label{eqn:corruptedCS2}
\end{equation}
with the new variable $w= [e,c]$ and the augmented matrix ${A}= [  \lambda^{-1} \, Id_m, \, \Phi_{m\times r}]$. This can be relaxed to:
\begin{equation}
\min_{w} \|w\|_1 \quad\text{subject to}\quad  ||y - Aw||_2 \leq \sigma,
\label{eqn:corruptedCS3}
\end{equation}
for some (non-negative) parameter $\sigma$. Using Equation~\eqref{eqn:corruptedCS3} gives the experiments more flexibility, while also coinciding with Equation~\eqref{eqn:corruptedCS2} when $\sigma$ is sent to zero. Following the derivation from \cite{schaeffer2018extracting}, %our cyclic paper
let $v$ be an auxiliary variable relaxing the constraints:
\begin{equation*}
(w,v)\in \mathcal{D}:=\left\{ (w,v) | \ v=Aw \right\} \quad  \text{and} \quad v\in B_\sigma(y) := \{ v \, | \ \|y-v\|_2 \leq \sigma \}.
\end{equation*}
Using the indicator function for a set $\mathcal{S}$:
  \begin{align*}
\mathbb{I}_{\mathcal{S}}(w):=
\begin{cases}
&\ \, 0, \ \ \text{if} \ \ w \in \mathcal{S}\\
 & \infty, \ \ \text{if} \ \ w \notin \mathcal{S},
 \end{cases}
  \end{align*}
Equation~\eqref{eqn:corruptedCS3} can be written as: \begin{equation}
\min_{(w,v)}\  g_1(w,v)+g_2(w,v), \label{eqn:DRobjectivefunction}
\end{equation}
where $g_1(w,v) :=  \| w \|_{1} + \mathbb{I}_{B_\sigma(y)}(v)$ and $g_2(w,v) := \mathbb{I}_{\mathcal{D}}(w,v)$.
The Douglas-Rachford algorithm uses the proximal operator within its iterations. The proximal operator of a function $g$ is defined as:
\[ 
\begin{aligned}
 \prox_{\gamma g}(z) :=  \argmin{x}  \left\{\frac{1}{2} \| z-x\|^2 + \gamma \, g(x)\right\},
 \end{aligned}\]
where $\gamma>0$ is a free parameter. As shown in \cite{schaeffer2018extracting}, the proximal operators for $g_1$ and $g_2$ are as follows:
\begin{align*}
 \prox_{\gamma g_1}(w,v) &= \left(  S_\gamma(w), \, \text{proj}_{B_\sigma(y)}(w) \right),
\end{align*}
where $\text{proj}_{B_\sigma}$ is the Euclidean projection onto the ball:
\begin{align*}
\text{proj}_{B_\sigma(y)}(v) :=
\begin{cases}
& \hspace{2.15cm} v, \ \ \text{if} \  ||v-y||_2 \leq \sigma\\
 & y +  \sigma\, \dfrac{v-y}{\|v-y\|_2}, \  \, \ \text{if} \ ||v-y||_2 > \sigma, \end{cases}
 \end{align*}
 and the soft-thresholding function $S$ (also known as \textit{shrink})   is defined component-wise as:
\begin{align*}
S_\gamma(w_j) = \text{sign}(w_j)\, \max\left(|w_j| - \gamma,0 \right). \end{align*}
The proximal operator for $g_2$ is defined by:
\begin{align*}
 \prox_{\gamma g_2}(w,v) &= \left( \, (Id_{m+r} + A^TA)^{-1}(w+A^T v), \ A(Id_{m+r} + A^TA)^{-1}(w+A^T v)\ \right).
\end{align*}
Define $\rprox_{\gamma g}(v) := 2 \prox_{\gamma g}(v)-v$, then the iterations for the Douglas-Rachford method are:
\begin{equation}
\begin{aligned}
(\tilde{w}^{k+1},\tilde{v}^{k+1})&=\frac{1}{2} \left( \ (\tilde{w}^{k},\tilde{v}^{k})+ \rprox_{\gamma g_2} \left(  \rprox_{\gamma g_1}\left(\tilde{w}^{k},\tilde{v}^{k}\right) \right) \ \right),\\
({w}^{k+1},{v}^{k+1})&=\prox_{\gamma g_1} (\tilde{w}^{k+1},\tilde{v}^{k+1}),
\end{aligned}\label{eq:DRalg}
\end{equation}
and will converge to a minimizer of Equation~\ref{eqn:corruptedCS3} for any $\gamma>0$ \cite{combettes2011proximal}.  Unless otherwise stated, in all of the computational examples we set $\gamma=1$ and $\sigma = 10^{-10}$.

%%%%%
\subsection{Computational Results}
%~~~~~~~~~~~~~~~~~~~~~~~~~~~~~~~~~~~~\\ 

Throughout the following examples, the uncorrupted data $\mathbf{x}^{(i)} \in \mathbb{R}^d$ are simulated as follows.
First, we simulate a sequence of bounded i.i.d. random variables $\mathbf{z}^{(i)}$, $i = 1, 2, \ldots m+3$.
Then, we set 
\[
\mathbf{x}^{(i)}  = \frac{\mathbf{z}^{(i)}}{16} + \frac{\mathbf{z}^{(i+1)}}{8} + \frac{\mathbf{z}^{(i+2)}}{4} + \frac{\mathbf{z}^{(i+3)}}{2}, \quad i = 1, 2, \ldots, m.
\]
It is readily to see that $(\mathbf{x}^{(i)})$ is a stationary exponentially strongly $\alpha$-mixing process.

\textbf{Example 1:} For the first example, we consider learning the function:
\begin{equation}
f(x) = 1 -2 x_1 x_2 x_3 + 5 x_1^5,
\label{eqn:example1}
\end{equation}
where $\mathbf{x} \in \mathbb{R}^3$. In Figure~\ref{fig:5thordereqn}, we display the probability of exact recovery of the polynomial coefficients in Equation~\eqref{eqn:example1} versus the sampling rate $\frac{m}{r}$, \textit{i.e.} the number of rows versus columns in $\Phi$. For this example, the matrix $\Phi$ contains all monomials up to fifth order in  $\mathbb{R}^3$ for a total of $56$ columns. We measure the probability of exact recovery by comparing the computed support set to the exact support set.  We compute the success rate over 100 trials with random sparse corruption, varying the sparsity of the corruption vector between $5$ (blue solid with dots), $10$ (red solid), and $12$ (yellow dotted). In all examples, the sparsity of the corruption vector refers to the number of non-zero elements, i.e. $s_\theta$. The threshold to achieve $90\%$ probability of success is marked by the horizontal dotted line. Note that the method is able to recovery the polynomial coefficients with $90\%$ probability even when the matrix $\Phi$ is under-sampled ($m<r$.) As the sparsity of the corruption vector increases, the probability of exact recovery decreases, as expected. For large enough sampling rates, in particular after 150 samples, the calculated probability of recovery is nearly one.

 \begin{figure}[t!]
\centering
 \includegraphics[width = 3.5 in]{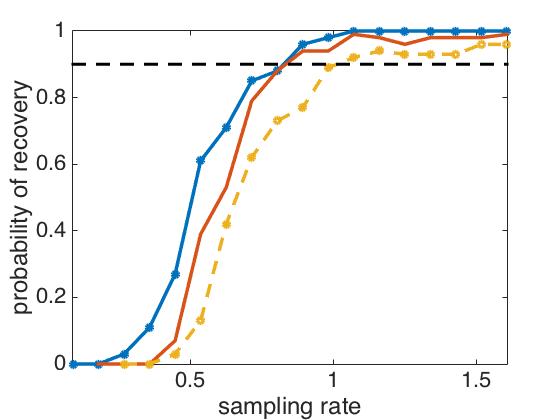}
\caption{Probability of exact recovery versus the sampling rate $\frac{m}{r}$ for the $5^{th}$ order equation \eqref{eqn:example1} with $d=3$ ($N = 56$ monomial terms) and the sparsity of the polynomial coefficient vector equal to $3$.  We measure the rate of exact recovery over 100 trials, varying the sparsity of the corruption vector: $5$ (blue solid with dots), $10$ (red solid), and $12$ (yellow dotted). The horizontal line is the threshold to achieve $90\%$ probability of success, which are all achieved with under-sampled data i.e. $m<r$. For large enough sampling rate (for example 150 samples), the probability of recovery approaches one.}
 \label{fig:5thordereqn}
 \end{figure} 

%%%%%
\textbf{Example 2:} 
The second example investigates the recovery of the function:
\begin{equation}
f(x) = -1 -2\, x_1^p,
\label{eqn:example2}
\end{equation}
for various $p>1$ where $\mathbf{x} \in \mathbb{R}^3$. In Table~\ref{Table1}, the probability of exact recovery versus the maximum degree of the polynomial is shown for two sampling rates: $15\%$ and $35\%$. The maximum degree of the monomials used in the dictionary is set to $10$ for all runs (for a total of $286$ monomial terms). The sparsity of the polynomial coefficient vector is $2$ and the sparsity of the corruption vector is set to $5$ in this example.  The corruption is drawn randomly for each trial.  For this example, the matrix $\Phi$ is normalized column-wise before applying the algorithm in order to prevent potential bias due to the column-scaling as $p$ increases. It was observed that normalization also helps with the numerical stability of the problem when using higher-order monomials.

\begin{table}[h!]
\caption{ Probability of exact recovery versus the maximum degree of polynomial in the unknown function (see Equation \eqref{eqn:example2}). The ambient dimension is $d=3$, the maximum degree of the candidate polynomials is $p= 10$ for all runs (thus $N=286$ monomial terms). The sparsity of the polynomial coefficients is equal to $2$ and the sparsity of the corruption vector is $5$. 
	We measure the rate of recovery over 100 trials for various $p$.} \label{Table1}
\begin{center}
 \begin{tabular}{|c||c|c|c|c|c|}\hline
  degree $p$: &2  &  3&  5 & 8 & 10    \\ \hline
  recovery, $15\%$ sampling:  & 83  &  95 &   85 &  82 & 79
  \\ \hline
  recovery, $35\%$ sampling:   & 98 &  98 &  99  &  99 & 99  \\ \hline
 \end{tabular}
\end{center}
\end{table}
For $15\%$ sampling (the second row of Table~\ref{Table1}), the computed probability of exact recovery is fairly stable (outside of the $p=3$ case). One can observe a small decrease in the recovery rate between $p=5$, $p=8$, and $p=10$. To test the stability in the high-recovery limit,  we set the sampling rate to $35\%$ (third row of Table~\ref{Table1}). For $35\%$ sampling,  as $p$ increases from $2$ to $10$ the recovery rate stays nearly the same ($98\%$ to $99\%$).  This would indicate that this approach is fairly stable to changes in the function, when the dictionary $\Phi$ is fixed.

%%%%%
\textbf{Example 3:} 
In this example, we investigate the recovery of the function:
\begin{equation}
f(x) = -8.5+ 9.6\,x_1\,x_4+ 0.3\, x_2\, x_5+ 5.7\, x_1^3+1.9 x_3\,x_9^2,
\label{eqn:example3}
\end{equation}
under various conditions on the corruption vector. The non-zero values of the corruption vector are sampled uniformly from $[-H,H]$. We compute the probability of exact recovery versus the values of $H$ and changes in the sparsity of the corruption vector.  The values of $H$ are chosen to match the range of the magnitudes of the coefficients in Equation~\eqref{eqn:example3}. For this example, we use all monomials up to degree three in  $\mathbb{R}^{10}$ for a total of $286$ monomial terms. The sampling rate is fixed at $17.5\%$ in order to highlight the variability between the recovery rates. In Table~\ref{Table2}, we vary $H$ between $0.5$, $2$, and $10$ and the sparsity of the corruption vector between $3$, $10$, and $15$. The scaling parameter $\lambda$ is set to $2$ for all experiments in Table~\ref{Table2}. In all cases, as $H$ increases, the recovery rate also increases. As the sparsity of the corruption vector increases, the recovery rate decreases, as expected. It was observed that the failures tended to occur on the vector $\mathbf{e}$, while the recovery of the polynomial coefficients were fairly stable to changes in $H$ and to changes in the sparsity of the corruption. In fact, it is possible, for certain cases, to recover the polynomial coefficients for relatively dense corruption vectors.    
\begin{table}[h!]
\caption{Probability of exact recovery versus the sparsity of the corruption vector (see Equation \eqref{eqn:example3}). The ambient dimension is $d=10$, the maximum degree of the candidate polynomials is $p=3$ (thus $N = 286$ monomial terms), and the sparsity of the polynomial coefficients is $5$.  The non-zero values in the corruption vector have uniform random values in $[-H,H]$. In each test, the sampling rate is $50\%$ and the parameter is set to $\lambda=2$. Note that tuning $\lambda$ may change the recovery rate.}\label{Table2}
\begin{center}
 \begin{tabular}{|c||c|c|c|}\hline
sparsity of the corruption &3 &10 &  15  \\ \hline
   recovery,  $H=0.5$ & 91   & 60 &  23    \\ \hline
  recovery,  $H=2$ & 98   &    65 &    28   \\ \hline
  recovery,  $H=10$ &  100 & 72  &       34
  \\ \hline
 \end{tabular}
\end{center}
\end{table}

%%%%%
\textbf{Example 4:} 
The last example details the recovery of the function:
\begin{equation}
f(x) = -1+2 \, x_1^2+0.5\,  x_5\,  x_{20},
\label{eqn:example4}
\end{equation}
when it is perturbed by the function $g(\mathbf{x}) = \epsilon \sin(2\pi x_1)$ (which is noise and is not part of the dictionary). We compute the probability of exact recovery of the polynomial coefficients in Equation~\eqref{eqn:example4} versus various values of $\epsilon$.  We use all monomials up to degree two in  $\mathbb{R}^{20}$ for a total of $231$ monomial terms. The sampling rate is fixed at $21.7\%$ and the sparsity of the corruption is set to $3$. In Table~\ref{Table3}, we can see that as $\epsilon$ increases, the recovery rate decreases. The average $\ell_1$ error (over the successful trials) is stable over the various tests.  For larger values of $\epsilon$, it may be possible to recover the sparse coefficients by adjusting $\lambda$. For example, by lowering $\lambda$ to $0.9$, the recovery rate for $\epsilon=10^{-3}$ becomes $94$ and the average $\ell_1$ error becomes $0.0141$.

\begin{table}[h!]
\caption{Probability of exact recovery (of the polynomial coefficients) versus the sparsity of the corruption vector (see Equation \eqref{eqn:example4}). The ambient dimension is $d=20$, the maximum degree of the candidate polynomials is $2$ (thus $N = 231$ monomial terms), and the sparsity of the polynomial coefficients is equal to $3$. In each test, the sampling rate is $21.7\%$ and the parameter is set to $\lambda=1$.}\label{Table3}
\begin{center}
 \begin{tabular}{|c||c|c|c|c|}\hline
$\epsilon$ &0 & $10^{-5}$ & $10^{-4}$ & $10^{-3}$     \\ \hline
   recovery            &99 &      98     &      98    &       85       \\ \hline
      $\ell_1$ error &0.0144 &  0.0120&0.0104    &     0.0156             \\ \hline
 \end{tabular}
\end{center}
\end{table}

\section{Conclusion}
Function approximation via $\ell_1$-optimization is a useful technique for automated learning. There are many results on the behavior of the $\ell_1$-solution when applied to i.i.d. data; however, theoretical results for dependent data is limited. The overall goal of this work is to show that under weaker conditions, exact and stable recovery is guaranteed. Specifically, we have shown that if the data is not independent but satisfies a suitable concentration inequality, one can provide a recovery guarantee for the learning function problem with corrupted data. Moreover, our proofs also show that the associated dictionary matrix generated from this type of data satisfies the null space property. It may be possible to weaken the requirements further while still preserving the core results. From numerical experiments, we observe that we need fewer number of measurements $m$ (compared to the theoretical bounds) to recover the underlying function. It is likely that new theories are needed to incorporate the sparsity level of the target function to relax the conditions on $m$. Specifically, it would be interesting to investigate the compressed sensing setting for dependent data where the number of measurements $m$ is less than $r$. In that direction, existing literature on the sparse linear regression problem in the compressive sensing setting considers only the ideal cases where the sampling matrix has i.i.d. Gaussian entries -- see for example  \cite{li2013compressed, xu2014system, foygel2014corrupted, price2018}, or is formed from bounded orthonormal systems \cite{adcock2017compressed}, randomly modulated unit-norm
frames, or randomly subsampled orthonormal matrices \cite{zhang2018uniform}. However, it is often the case that sparsity (or sparsity with respect to some nice bounded orthonormal basis) may not hold. One of the benefits of our results is that they hold when the target function is indeed sparse with respect to monomials, approximately sparse (where sparsity helps with overfitting), or even dense.

\section*{Acknowledgments} 
 L. S. T. H. acknowledges the support of startup funds from Dalhousie University, the Canada Research Chairs program, and NSERC Discovery Grant. 
H. S. acknowledges the support of AFOSR, FA9550-17-1-0125 and the support of NSF CAREER grant $\#1752116$. 
G. T. acknowledges the support of NSERC Discovery Grant. 
R. W. acknowledges the support of NSF CAREER grant $\#1255631$.

%\bibliographystyle{plain}
%\bibliography{HSTW}

\end{document}